\documentclass{llncs}

\usepackage[utf8]{inputenc}
\usepackage[T1]{fontenc}
\usepackage[scaled=0.8]{beramono}

\usepackage{amsmath}
\usepackage{amssymb}
\usepackage{stmaryrd}
\usepackage{mathtools}
\usepackage[shortlabels]{enumitem}
\usepackage{wrapfig}
\usepackage{tabto}
\usepackage{booktabs}
\usepackage{comment}
\usepackage{todonotes}

\usepackage[allcolors=blue,colorlinks=true]{hyperref}

\usepackage{proof}
\usepackage{cleveref}
\usepackage[numbers,sort&compress]{natbib}

\newcommand{\code}[1]{\text{\upshape\ttfamily #1}}

\def\chop  {\mathrel{\raise 0.6ex\hbox{\oalign{\hfil$\scriptscriptstyle
    \mathrm{o}$\hfil\cr\hfil$\scriptscriptstyle\mathrm{9}$\hfil}}}}

\renewcommand{\Relbar}{\mathrel{\mkern-0.5mu=\mkern-0.5mu}}
\newcommand{\sledom}{\Relbar\joinrel\mathrel{|}}
\MakeRobust{\Relbar}
\MakeRobust{\sledom}

\newcommand{\hoare}[3]{\{\,#1\,\}~#2~\{\,#3\,\}}
\newcommand{\pre}[2]{\mathit{wp}(#1 \mid #2)}
\newcommand{\Pre}[2]{\mathit{wp}\bigl(#1 \mid #2\bigr)}

\newcommand{\wpre}{\mathit{wp}}

\newcommand{\eoe}{\hfill\scalebox{0.8}{$\blacksquare$}}

\renewcommand{\vec}[1]{\overline{#1}}
\newcommand{\vx}{\vec{x}}
\newcommand{\vy}{\vec{y}}

\newcommand{\ve}{\vec{e}}

\renewcommand{\phi}{\varphi}

\newcommand{\free}{\mathit{free}}
\newcommand{\vars}{\mathit{vars}}
\renewcommand{\mod}{\mathit{mod}}
\newcommand{\WHILE}{\code{while}}

\newcommand{\IF}{\code{if}}
\newcommand{\ASSUME}{\code{assume}}
\newcommand{\HAVOC}{\code{havoc}}
\newcommand{\THEN}{\code{then}}
\newcommand{\ELSE}{\code{else}}
\newcommand{\SKIP}{\code{skip}}

\newcommand{\isPrime}{\mathit{prime}}
\newcommand{\isOdd}{\mathit{odd}}

\newcommand{\Nil}{\code{[\kern 0.2pt]}}
\newcommand{\Cons}{\mathbin{\code{:\kern -2.5pt:}}} % changed from : to :: to avoid confusion with typing

\newcommand{\gray}[1]{{\color{gray}#1}}

\usepackage{newunicodechar}
\newunicodechar{₀}{\ensuremath{_0}}
\newunicodechar{₁}{\ensuremath{_1}}
\newunicodechar{₂}{\ensuremath{_2}}
\newunicodechar{₃}{\ensuremath{_3}}
\newunicodechar{₄}{\ensuremath{_4}}
\newunicodechar{₅}{\ensuremath{_5}}
\newunicodechar{₆}{\ensuremath{_6}}
\newunicodechar{₇}{\ensuremath{_7}}
\newunicodechar{₈}{\ensuremath{_8}}
\newunicodechar{₉}{\ensuremath{_9}}

\newcommand{\nxt}{\circ\,}
\newcommand{\eventually}{\lozenge\,}
\newcommand{\always}{\Box\,}
\newcommand{\until}{\mathrel{\mathcal{U}}}

\newcommand{\wuntil}{\mathrel{\mathcal{W}}}
\newcommand{\xwuntil}{\mathrel{\mathcal{W}^+\!}}
\newcommand{\release}{\mathrel{\mathcal{R}}}
\newcommand{\xrelease}{\mathrel{\mathcal{R}^+\!}}

\newcommand{\unfold}{\mathit{unfold}}
\newcommand{\step}{\mathit{step}}

\newcommand{\ITER}[1]{\code{iter}(#1)}

\newcommand{\measure}{\mathcal{M}}
\newcommand{\hypothesis}{\mathcal{H}}
\newcommand{\induct}{\mathcal{I}}
\newcommand{\coinduct}{\mathcal{C}}

\newtheorem{algorithm}{Algorithm}
\crefname{algorithm}{algorithm}{algorithms}
\Crefname{algorithm}{Algorithm}{Algorithms}

\newtheorem{axiom}{Axiom}
\crefname{axiom}{axiom}{axioms}
\Crefname{axiom}{Axiom}{Axioms}

\author{Gidon Ernst\\\email{gidon.ernst@uni-a.de}}
\institute{University of Augsburg, Augsburg, Germany \\ LMU Munich, Munich, Germany}

\begin{document}

\title{Weakest Precondition Rules for \\ Programs with Linear Temporal Specifications}

\maketitle

\begin{abstract}
With today's mature auto-active program verification tools
complex functional requirements can be formalized and proved.
To that end, they rely on verification condition generation
to bridge between structured programs and high-level specifications
and the automated theorem provers used in the background.
Integrating software modules into larger systems may necessitate
to consider temporal logic requirements, notably liveness properties over infinite traces.
Unfortunately, most state-of-the-art tools lack explicit support for such temporal specifications.
There are various proposals that address the integration of structured programs and temporal logic, but each comes with some inherent limitation regarding expressiveness or automation.
In this paper, we demonstrate a simple but universal solution
that can be integrated easily into existing verification condition generators.
\end{abstract}

\section{Introduction}

Formal methods aim to assert the quality of software systems with respect to requirements
by formalizing both within a suitable mathematical language
and then establishing a rigorous correspondence between the two.
Depending on the view on the system and the properties of interest,
formal methods cover a wide spectrum of possibilities.

A key challenge is to \emph{bridge across} formal paradigms.
Due to the disparity of design points,
the respective foundations, methods, and tools are not necessarily integrated or even compatible.
Recently, there have been a series of efforts~\cite{ernst:dagstuhl22451} and events to bring together
the high-level behavioral view of systems and their components
with a more data-centric and algorithmic view of software implementations.
One of the goals of the discussions was to clarify what is missing to progress from the higher-level systems view,
where specifications typically expressed in temporal logic over global execution traces,
down to the software perspective,
in which characterizing the properties of individual execution states is usually the main challenge.

As one of the gaps, it was identified that while model-checking
of temporal properties for is commonplace and well understood~\cite{clarke2018handbook},
the same cannot be said when it comes to auto-active
and deductive verificaton approaches for structured programs.
It is possible to encode temporal specificationsin these tools,
but that is error-prone, cumbersome and lacks integration
with the underlying verification condition generation.
Conversely, TLA$^+$~\cite{lamport2002} does have first-class support for temporal properties,
but then the programs need to be encoded into transition relations,
leading to analogous draw-backs.

To construct a program logic for temporal specifications,
several challenges should be addressed specifically:
Executions are potentially infinite,
which suggests co-inductive foundations~\cite{nakata2015hoare}.
Sequential composition of programs needs to be connected with linear traces,
but the straight-forward compositional rule used in~\cite{nakata2015hoare,gurov2024expressive} can not be automated algorithmically.
Dynamic Logic with traces~\cite{beckert2013dynamic} recovers Disjkstra's
elegant approach to sequential composition,
but the approach is limited to non-nested temporal operators.
Reasoning about iteration likewise becomes more complicated,
necessitating more powerful (co-)induction rules~\cite{ioannidis2025structural,schellhorn2014rgitl}
and arguments beyond simple invariants,
which can neatly be captured by introducing continuation parameters
in the correctness judgements~\cite{ioannidis2025structural,gurov2024expressive}.

The \textbf{contribution} of this paper is to assemble the individual ideas found in the literature
into a system that is straight-forward
and for which a verification condition generator can be implemented easily.
As a consequence, the approach is fully compatible with the auto-active verification paradigm,
in which critical proof hints are annotated onto the program source code,
so that verification conditions become implications in the underlying assertion language for which we assume a proof oracle.

The \textbf{scope} of this paper is to present the calculus
and to show-case its functionality on examples.
The approach is sound and complete relative
to the background logic and the supported
induction/coinduction principles associated to temporal operators
and the use of well-founded measures.
We leave it for future work to implement this approach in a mature deductive verification tool and to conduct larger experiments.
Similarly, we leave extensions like concurrency or communication to future work.

\paragraph{Data Availability Statement.}
The formalization of this paper and the correctness of all rules
has been mechanized in Isabelle/HOL.
A verification condition generator for this calculus has been implemented in Scala.

The artifact for an earlier version of this paper~\cite{ernst2026weakestpreconditioncalculusprograms}
is available at \url{https://doi.org/10.5281/zenodo.18427930},
an updated version will be provided shortly.
    % \todo{The artifact will appear over the coming days.}

\section{Preliminaries and Motivation}
\label{sec:preliminaries}

The presentation of this paper discerns imperative \emph{program commands}~$c$
that encompass atomic actions~$A$, sequential composition,
nondeterministic choice, and non-determinstic iteration.
Typical atomic actions include deterministic and nondeterministic assignment
and assumptions.
\begin{flalign*}
\qquad
\textbf{commands} \quad
    c &\Coloneqq A \mid c_1;c_2 \mid c_1 \sqcup c_2 \mid \ITER{c}
    \mid \dots
    && \\
\textbf{actions} \quad
    A &\Coloneqq \SKIP \mid \vy \coloneqq \ve \mid \HAVOC~\vy \mid \ASSUME~P \mid \dots
\end{flalign*}
A suitable and standard interpretation of commands is in terms of finite traces for terminating executions
and infinite traces for non-terminating executions.
Sequential composition is interpreted by the ``chop'' operator,
see e.g.~\cite{itl2015,schellhorn2014rgitl}.
Constructs like~$\IF\text{-}\THEN\text{-}\ELSE$ and~$\WHILE$ loops can be derived as usual.

We consider the assertion language of LTL with the standard operators,
where~$P,Q$ stand for state formulas,
$\odot \in \{ \land, \lor, \Rightarrow, \Leftrightarrow \}$ are the binary connectives,
and $\mathcal{Q} \in \{ \exists, \forall \}$ are the quantifiers.
\begin{flalign*}
\qquad
    \textbf{formulas} \quad
\phi,\psi &\Coloneqq
    P \mid \nxt \phi \mid \eventually \phi \mid \always \phi \mid
    \phi \until \psi \mid
    \phi \release \psi \mid
    \phi \wuntil \psi \mid \cdots
    \\
    & \hspace*{0.44cm}
    \lnot \phi \mid
    \phi \odot \psi \mid
    \mathcal{Q}\ \vx.\ \phi \mid \cdots
    &&
\end{flalign*}
A suitable and standard interpretation is over infinite traces.
We expect that quantifiers determine the values of bound variable
throughout the entire trace, and not just in the initial state.

To bridge between commands~$c$,
which may have finite or infinite executions,
and formulas,
which are evaluated over infinite traces only,
we never view a command in isolation.
Instead, we take a ``continuation-passing'' view, wherein a sequence of commands
must always be followed by a formula~$\phi$
or a placeholder~$\omega(\vx)$ that denotes the remainder trace.
Programs in this work therefore are formed according to the following grammar, and we say that~$\kappa$ in~$c;\kappa$ is the \emph{continuation} of command~$c$.
\begin{flalign*}
\qquad
    \textbf{programs} \quad
\kappa &\Coloneqq
    \phi \mid \omega(\vx) \mid c;\kappa
    &&
\end{flalign*}
For example, $c;\mathit{false}$ encodes that we are interested in non-terminating runs of command~$c$ only,
whereas $c;\ITER{\SKIP};\mathit{false}$ adds a stuttering loop to~$c$
that indefinitely repeats the final state.
This idea mimicks the idea used in model-checking to complete automata by additional transitions.

Placeholders~$\omega(\vx)$ serve as a mechanism for abstraction.
The idea is that fresh identifiers for~$\omega$ can be generated on-demand,
but since they remain abstract, the variables~$\vx'$ on which they depend must be mentioned explicitly.

We collectively refer to commands, formulas, and programs as ``trace properties''~$\tau$.
We denote by $\vars(\kappa)$ the program variables occurring in~$\kappa$
and by~$\mod(\kappa)$ the program variables potentially assigned to by~$\kappa$.
We denote by $\free(\phi)$ the free variables occurring in formula~$\phi$.
When~$\vx'$ are fresh copies of some variables~$\vx$,
we denote the renaming of~$\vx$ to~$\vx'$ by ascribing a prime symbol
as in~$c'$, $\phi'$, and~$\kappa'$, respectively.
We write $\phi \models \psi$ when~$\phi$ implies~$\psi$ over all traces
and we write~$\phi \equiv \psi$ when $\phi \models \psi$ and $\psi \models \phi$ (similarly for other trace properties).

\medskip

Correctness as captured by Hoare logic can be unerstood
as establishing a \emph{contract} for a command~$c$,
where a triple $\hoare{P}{c}{Q}$ expresses precondition~$P$ can be assumed in the pre-state by~$c$,
which in turn has to guarantee postcondition~$Q$ in any final state.
For this work, the idea is that programs is are not running in isolation,
but as part of an environment whose behavior is temporally intertwined with the execution steps.
The contracts are therefore concerned not just with initial and final states,
but rather with entire traces,
for which the environment makes certain \emph{assumptions}~$\alpha$
and the program in turn provides certain \emph{guarantees}~$\gamma$,
which together form a temporal logic contract.
This motivates the following definitions of correctness.
\begin{definition}[Contract]
    \label{def:contract}
A contract $\psi = (\alpha \Rightarrow \gamma)$
consisting of
a temporal logic assumption~$\alpha$
and a temporal logic guarantee~$\gamma$.
\end{definition}
\begin{definition}[Correctness Judgements]
    \label{def:hoare}
For a contract $\alpha \Rightarrow \gamma$
consisting of
a temporal logic assumption~$\alpha$
and a temporal logic guarantee~$\gamma$,
program~$c;\kappa$ is correct,
if all traces produced in an environment satisfying~$\alpha$ imply~$\gamma$.
\begin{align*}
\hoare{\alpha}{c;\kappa}{\gamma}
    \qquad \text{ iff } \qquad
    (c;\kappa) \models (\alpha \Rightarrow \gamma)
\end{align*}
\end{definition}
Note, in contrast to Hoare logic, where the pre- and postcondition are evaluated in different states,
here the distinction into assumptions and guarantees is a conceptual one,
i.e., $\hoare{\alpha}{c;\kappa}{\gamma}$ can be equivalently
expressed as $\hoare{\mathit{true}}{c;\kappa}{\alpha \Rightarrow \gamma}$
or even $\hoare{\alpha\land\lnot\gamma}{c;\kappa}{\mathit{false}}$.
Some readers might find it useful to think of~$\alpha \Rightarrow \gamma$
as a \emph{sequent} that is embedded into the correctness judgement
but which can be manipulated independently from the program.

\Cref{def:hoare} is sufficiently general to
encode state-based Hoare logic
as well as the correctness judgements of related work like~\cite{nakata2015hoare,beckert2013dynamic,gurov2024expressive,schellhorn2014rgitl} (see \cref{sec:related}).

\begin{example}[Specification of a Prime Number Generator]
    \label{ex:primes}
A program that repeatedly increments~$x$ will produce an infinite number of primes:
\begin{align}
\hoare{x \ge 0}{\ITER{x \coloneqq x+1};\mathit{false}}{\always\eventually\isPrime(x)}
    \label{eq:primes}
\end{align}
This program is intuitively correct, because the distance from one prime to the next is bounded.
Formally, the proof depends on a layered inductive argument,
where the main mechanism of repetition is justified from the $\always$-operator
and progress from one prime to the next is justified from well-founded induction over a function~$\delta(x)$ that measures the distance from~$x$ to the next prime.
    \eoe
\end{example}

\section{Weakest Precondition Calculus}
\label{sec:wp}

In this section, we show how to analyze correctness judgements
$\hoare{\alpha}{c;\kappa}{\gamma}$
using a weakest-precondition approach to verification condition generation:
\begin{definition}[Weakest Precondition]
    \label{def:wp}
The weakest precondition of program $\kappa$
with respect to contract~$\psi = (\alpha \Rightarrow \gamma)$,
is the weakest formula~$\alpha'$ so that
$\hoare{\alpha \land \alpha'}{\kappa}{\gamma}$.
This can be expressed equivalently as
$\alpha' \equiv (\kappa \Rightarrow \psi)$.
\end{definition}
We are interested in calculating a plain temporal logic formula~$\pre{\kappa}{\psi}$ that is a weakest precondition for~$\kappa$ and~$\psi$.
\begin{algorithm}[Verification Conditions]
    \label{alg:wp}
The weakest precondition is computed by structural recursion
on the program for an arbitrary~$\psi$:
\begin{align}
\pre{\phi}{\psi}
    & \iff (\phi \Rightarrow \psi)
    \label{wp:stop}
    \\
\pre{(c_1 \sqcup c_2);\kappa}{\psi}
    & \iff   \pre{c_1;\kappa}{\psi}
      \land \pre{c_2;\kappa}{\psi}
    \label{wp:choice}
    \\
\pre{(c_1; c_2);\kappa}{\psi}
    & \iff \pre{c_1; (c_2;\kappa)}{\psi}
    \label{wp:seq}
    \\
\pre{A;\kappa}{\psi}
    & \iff \big(
        \forall\ \vx'.\ \step(A,\vx,\vx')
            \Rightarrow \pre{\kappa'}{\unfold(\psi,\vx,\vx')}
        \big)
    \label{wp:atom}
    \\[4pt]
\pre{\ITER{c};\kappa}{\psi}
    & \iff \text{exists formula } \iota,
      \text{ closure operator } \hypothesis_\iota \text{ for } \iota,
         \label{wp:iter}
         \\
    & \hspace{30pt} \text{and fresh placeholder } \omega \text{ so that}
          \nonumber \\
    & \hspace{30pt}
        (\iota \Rightarrow \psi)
        \land \forall\ \vx.\
            \pre{\kappa}{\iota} \land
                \pre{c;\omega(\vx)}{\hypothesis_\iota(\omega(\vx)) \Rightarrow \iota}
         \nonumber
\end{align}
\end{algorithm}
When there is no leading program command,
we just unfold it into a temporal implication~\eqref{wp:stop}.
The rule for choice~\eqref{wp:choice} splits apart both branches.

Thanks to continuations being part of programs and thanks to associativity of sequential composition
rule~\eqref{wp:seq} gradually exposes the leading atomic command~$A$.
In contrast to vanilla Hoare logic and approaches like~\cite{gurov2024expressive,nakata2015hoare}, there is no need to come up with an intermediate assertion.

When the leading statement is an atomic action~$A$,
rule~\eqref{wp:atom} executes it through its next-state relation $\step(A,\vx,\vx')$.
It relies on a helper function~$\unfold$,
described below, which re-arranges the contract~$\psi$
into current-state and next-state components, whereas the latter are expressed with respect to~$\vx'$.
The rule is reminiscent to how Dijkstra's calculus~\cite{dijkstra1975guarded}
translates the postcondition~$Q$ in~$\wpre(c,Q)$
into the current state by introducing fresh logical variables corresponding to different time points during the execution.
The intuition is that~$\vx'$ copies the values of~$\vx$ with an offset of one step.

Rule~\eqref{wp:iter} for iterations mimicks the structure of invariant-based loop verification.
It introduces a generalization~$\iota$ of the current contract~$\psi$
and produces two verification conditions, one for the exit case,
in which the continuation~$\kappa$ of the loop on its own must establish the contract~$\iota$,
and one for the step case,
in which an arbitrary iteration of the loop body~$c$ is analyzed.
The unviversal quantifier over the modified program variables~$\vx = \mod(c)$ generalizes over arbitrary starting states as usual.
$\hypothesis_\iota$ is a \emph{hypothesis} specific to~$\iota$, as formalized in \cref{def:hypothesis},
whose purpose is to represent the (co-)induction principles of our choice.
It strengthens guarantee~$\iota$ by adding additional assumptions.
Hence, the role of~$\hypothesis_\iota$ is to encode \emph{what} information becomes available at the end of an iteration.
Placeholder $\omega$ that,
by occurring as the continuation of the iteration of~$c$,
encodes \emph{when} the information encoded into the hypothesis generated by~$\hypothesis_\iota$ becomes available.
In practice, we expect the choice of~$\iota$ and~$\hypothesis_\iota$
to be annotated into the program, which is a common strategy
in auto-active verificaiton tools.

\begin{algorithm}[Temporal Unfolding]
    \label{alg:unfold}
The transformation~$\unfold(\psi,\vx,\vx')$ splits
a contract~$\psi$ into its current-state constituents
and its residual next-state guarantees,
in which we replace~$\vx$ by~$\vx'$ and shift them to the current state.
\normalfont
\begin{align*}
\unfold(p)
    & ~=~ p
        \hspace{33.5pt} \text{ keep } \vx
        %\label{eq:unfold-prop}
    \\
\unfold(\nxt \phi,\vx,\vx')
    & ~=~ \phi'
        \hspace{30pt} \text{ replace } \vx \text{ by } \vx'
        %\label{eq:unfold-nxt}
    \\
\unfold(\always\phi,\vx,\vx')
    & ~=~ \unfold(\phi,\vx,\vx') \land \always\phi'
        %\label{eq:unfold-always}
    \\
\unfold(\eventually\phi,\vx,\vx')
    & ~=~ \unfold(\phi,\vx,\vx') \lor \eventually\phi'
        %\label{eq:unfold-eventually}
    \\[4pt]
\unfold(\lnot \phi,\vx,\vx')
    & ~=~ \lnot \unfold(\phi,\vx,\vx')
    \\
\unfold(\phi \odot \psi,\vx,\vx')
    & ~=~ \unfold(\phi,\vx,\vx') \odot \unfold(\psi,\vx,\vx')
        \quad \text{for } \odot \in \{\land,\lor,\Rightarrow,\Leftrightarrow\}
    \\
\unfold(\mathcal{Q}\ \vy.\ \phi,\vx,\vx')
    &\iff \mathcal{Q}\ \vy\,\vy'.\ \unfold(\phi,\vx\,\vy,\vx'\vy')
    \\
    & ~\dots
\end{align*}
\end{algorithm}
The algorithm follows the structure of the formula
and uses the unfolding laws of temporal operators
like $\always \phi \Leftrightarrow \phi \land \nxt \always \phi$.
Additionally, subformulas with a leading $\nxt\_$ operator are ``shifted''
into the current state with respect to the variables~$\vx'$.
Quantifiers introduce primed copies of the bound variables.

\begin{example}[Temporal Unfolding]
For the contract from \cref{ex:primes} 
\begin{align*}
& \unfold(x \ge 0 \Rightarrow \always\eventually\isPrime(x),x,x') = \\
& \quad x \ge 0
    \Rightarrow
        \big( \isPrime(x) \lor \eventually\isPrime(x') \big)
            \land
        \big( \always\eventually\isPrime(x') \big)
\end{align*}
\end{example}
To reason about loops,
hypotheses
uniformly describe the induction and coinduction principles
associated to temporal operators and for well-founded orders.
\begin{definition}[Hypothesis]
    \label{def:hypothesis}
A function~$\hypothesis_\iota(\_)$ from trace properties to trace properties
is a \emph{hypothesis} for a formula~$\iota$ if
\begin{align*}
\tau \models \hypothesis_\iota(\tau) \Rightarrow \iota
    \quad \text{ implies } \quad
\tau \models \iota
    \quad \text{ for all trace properties } \tau
\end{align*}
\end{definition}
\begin{lemma}[Well-founded Induction]
    \label{lem:wf}
For a well-founded measure~$\delta$,
function
$\measure^\delta_\iota(\tau)$
is a hypothesis that is independent of the shape of~$\iota$.
\begin{align*}
\measure^\delta_\iota(\tau)
    ~\coloneqq~
    \exists\ z.\ z = \delta
        \land
            \always (\delta < z \land \tau \Rightarrow \iota)
\end{align*}
\end{lemma}
The auxiliary variable~$z$ captures the value of~$\delta$ with respect to the current state.
In any (future) state in which~$\delta$ has decreased properly,
we can appeal to the hypothesis to generate the fact that~$\iota$ holds again
for as long as we have established the additional assumptions~$\tau$, too.

More interestingly, we can derive an induction principle from an environment assumption in a contract
$\eventually \phi \Rightarrow \gamma$
and a coinduction principle for a guarantee in a contract
$\alpha \Rightarrow \always \phi$.
\begin{lemma}[Induction, Coinduction]
    \label{lem:ci}
For the formulas~$\iota$ of the respective shape,
the following functions~$\hypothesis_\iota$ are hypotheses for arbitrary~$\phi,\alpha,\gamma$.
\begin{align*}
\induct_{\eventually \phi \Rightarrow \gamma}(\tau)
    & ~\coloneqq~
        \phantom{\lnot} ~\phi \xrelease (\tau \Rightarrow \gamma)
    \\
\coinduct_{\alpha \Rightarrow \always \phi}(\tau)
    & ~\coloneqq~
        \lnot \big(\phi \xwuntil (\alpha \land \tau) \big)
\end{align*}
\end{lemma}
The negation in front of the coinduction principle~$\coinduct$
comes from the fact that hypotheses are added as assumptions in~\eqref{wp:iter},
whereas~$\always \phi$ occurs in the conclusion of the guarantee.
Adding the negation compensates for this switch of polarity.

In this lemma, we make use of variants of the ``weak until'' and ``(weak) release'' operators from LTL,
which expose property~$\phi$ for at least one step proper before recurrence is allowed.
The second construction to prove a $\always$-property closely follows an idea shown in~\cite{beckert2013dynamic}, wherein it is expressed using~$\wuntil$ and~$\nxt$.
The appraoch for~$\eventually$ is dual
it it is not difficult to define similar principles for other LTL operators like~$\until$ and~$\wuntil$.

Both operators $\xrelease$ and $\xwuntil$ can be defined as abbreviations.
For this discussion what matters are the following fixpoint characterization,
which justify \cref{lem:ci},
and unfolding laws, which are used by \cref{alg:unfold}.
\begin{align}
\eventually \phi
    & \iff \mu\, \tau.\
            \phi \xrelease \tau
    && \text{where } &
            \phi \xrelease \psi
                & \iff \phi \lor \nxt (\psi \land (\phi \xrelease \psi))
    \label{eq:FR-unfold}
    \\
\always \phi
    & \iff \nu\, \tau.\
            \phi \xwuntil \tau
    && \text{where } &
            \phi \xwuntil \psi
                & \iff \phi \land \nxt (\psi \lor (\phi \xwuntil \psi))
    \label{eq:GW-unfold}
\end{align}
This formulation allows one to take \emph{multiple} steps before recurring to~$\tau$ in the fixpoint characterizations of~$\always\phi$ and~$\eventually\phi$.
\begin{example}[Verification conditions for a simple loop]
    \label{ex:odd}
We aim to calculate the weakest precondition
of a loop with two assignments for the guarantee that~$x$ is always an odd number.
\begin{align*}
\pre{\ITER{x \coloneqq x + 2; x \coloneqq x + 4};\mathit{false}}
    {\always \underline{\isOdd(x)}}
\end{align*}
Since the leading statement is a loop we apply rule~\eqref{wp:iter}.
The loop clearly preserves~$\isOdd(x)$
so we take~$\iota \equiv \isOdd(x) \Rightarrow \always \underline{\isOdd(x)}$,
wherein the premise $\isOdd(x)$ takes the role of a state-based invariant.
The side condition that $\iota$ implies the desired guarantee
will leave a residual constraint $\isOdd(x)$ as the weakest precondtion,
so far so good.
The loop is abstracted by choosing the hypothesis~$\coinduct_\iota(\omega(x))$
from~\cref{lem:ci}
which yields two conjuncts for fresh variables~$\vx'$ and placeholder~$\omega$:
\begin{align*}
& \text{base case} &
&\Pre{\mathit{false}}{\iota}
    \stackrel{\eqref{wp:atom}}{\iff}
        (\mathit{false} \Rightarrow \iota)
        \tag{\checkmark}
    \\[4pt]
& \text{step case} &
\forall x.\ &\Pre{x \coloneqq x + 2; x \coloneqq x + 4;\omega(x)~}
     {~\coinduct_\iota(\omega(x)) \Rightarrow \iota~}
\end{align*}
The base case is vacuous, as we have assumed the loop does not exit in the first place.
Let's expand the compound guarantee of the step case
\begin{align*}
\lnot \big(~
    \underbrace{\underline{\isOdd(x)} \xwuntil (\isOdd(x) \land \omega(x)) \big)}_\text{hypothesis}
    \big)
    \Rightarrow
        \big(\underbrace{~\isOdd(x)\phantom{\underline{(}}}_\text{invariant}
             \Rightarrow 
                 \underbrace{~\always \underline{\isOdd(x)}~}_\text{original guarantee} \big)
\end{align*}
which is equivalent to
\begin{align}
    \underbrace{~\isOdd(x)\phantom{\underline{(}}}_\text{invariant}
    \Rightarrow
            \underbrace{~\underline{\isOdd(x)} \xwuntil (\isOdd(x) \land \omega(x)) ~}_\text{coinductive goal}
             \lor
                 \underbrace{~\always \underline{\isOdd(x)}~}_\text{original guarantee}
        \label{eq:odd}
\end{align}
The coinductive goal generated by~$\coinduct_\iota$ now represents an \emph{alternative} goal,
in addition to the original guarantee,
which we can target on those executions of the loop body which terminate and eventually make~$\omega(x)$ true.
Including the original guarantee still leaves the option to apply a different proof strategy to those executions which diverge.

The weakest precondition of the loop body is calculated by applying rule~\eqref{wp:atom} twice and finally rule~\eqref{wp:stop},
which introduces two fresh variables, denoted~$x'$ and~$x''$.
Guarantee~\eqref{eq:odd} is unfolded using \cref{alg:unfold} and \eqref{eq:GW-unfold}.
Observe that the rule for atomic actions~\eqref{wp:atom} never pulls
state-formulas from the guarantee out of the weakest precondition operator.
Therefore, assumptions about the trace are neatly separated from the desired conclusions.
\begin{align*}
& \isOdd(x) \land x' = x + 2 \land x'' = x + 4 \land \omega(x'')
    \implies \\[8pt]
& \qquad \underline{\isOdd(x)} \land
    \Big( \underbrace{\isOdd(x') \land \omega(x')}_\text{unavailable} ~~\lor~~
  \underline{\isOdd(x')} \land
     \big(  \underbrace{\isOdd(x'') \land \omega(x'')}_\text{(\checkmark)} {} \lor \dots \big) \Big) \\
& \qquad  \gray{{} \lor {} \underbrace{~\underline{\isOdd(x)} \land \underline{\isOdd(x')} \land \always \underline{\isOdd(x'')}~}_\text{original}}
\end{align*}
For the first step, unfolding~$\xwuntil$ requires us to prove the underlined guarantee at least once, for~$x$.
After the first step we have the first opportunity to recur,
by proving~$\isOdd(x') \land \omega(x')$,
but while the invariant holds,
the placeholder~$\omega$ that signifies the end of the iteration
is rightfully not yet available for~$x'$.
Hence, we have to prove the underlined guarantee a second time, for~$x'$.
After the second step, $\omega(x'')$ becomes available from the continuationand since invariant~$\isOdd(x'')$ holds, too, we can conclude the proof~(\checkmark).
The original guarantee (grey) is unfolded, too, and kept around but remains irrelevant in this example.
    \eoe
\end{example}

\begin{lemma}[Compositionality of Hypotheses]
Given two hypotheses~$\hypothesis_1$ and~$\hypothesis_2$ for~$\iota$,
their composition $\hypothesis_1 \triangleright \hypothesis_2$ is again a hypothesis for~$\iota$
\begin{align*}
(\hypothesis_1 \triangleright \hypothesis_2)(\tau)
    ~\coloneqq~ \hypothesis_1(\tau) \land
                \hypothesis_2\big(\tau \land \hypothesis_1(\tau)\big)
\end{align*}
\end{lemma}
Note, the operator is not symmetric:
$\hypothesis_2$ is applied relative to the information that~$\hypothesis_1$ provides.
The construction resembles inductions over lexical combinations
of well-founded orders.
In some cases, including~$\hypothesis_1$ inside~$\hypothesis_2$ may not be necessary,
for example, when~$\hypothesis_1(\tau) \Rightarrow \always \hypothesis_1(\tau)$ is persistent over execution steps.
We leave such optimizations for future work.

\begin{example}[Verification conditions for a Prime Number Generator]
We sketch the verification conditions for \cref{ex:primes}.
\begin{align*}
\pre{\ITER{x \coloneqq x+1};\mathit{false}}{\iota}
    \quad \text{ for } \quad \iota = \big( x \ge 0 \Rightarrow \always\eventually\isPrime(x) \big)
\end{align*}
Proof structure is provided by \cref{lem:ci} for setting up the main mechanism for repetition via the outer~$\always$ operator.
A nested induction is by \cref{lem:wf} for a measure~$\delta(x)$
that counts the distance to some prime larger than~$x$.
The proof is then based on the composition~$\hypothesis \coloneqq \coinduct_\iota \triangleright \measure_\iota^\delta$.

The core of the verification conditions contains in its premise
the assumptions about the invariant at loop head,
the equation produced by the assignment,
the continuation~$\omega(x')$ (first line)
as well as the assumption added by the well-founded induction principle (second line, omitting the auxiliary~$z$).
\begin{align*}
& x \ge 0 \land x' = x + 1 \land \omega(x') \land
    \measure_\iota\big(\omega(x) \land \coinduct_\iota(\omega(x))\big)
        \implies \\
& \qquad
    \underbrace{
        \big(~\underbrace{\isPrime(x) \lor \eventually \isPrime(x')}_{\unfold(\eventually \isPrime(x), x,x')} \big)
                \land \big( (\underbrace{~x' \ge 0 \land \omega(x')~}_\text{\checkmark}) \lor \dots \big) }_{
        \unfold((\eventually \isPrime(x)) \xwuntil x (\ge 0 \land \omega(x)))
    }
\end{align*}
Similarly to \cref{ex:odd}, we present the coinductive goal in the conclusion of the verification condition.
We must demonstrate that at least one prime will be produced,
either in the current state for~$x$ or later for~$x'$.
The remainder of the executions are covered (\checkmark)
by re-establishing the invariant and by relying on continuation~$\omega(x')$.

If we are lucky, $x$~is indeed prime.
Otherwise, the following property of~$\delta$
\begin{align}
\lnot \isPrime(x) \implies \delta(x+1) < \delta(x)
    \label{lem:primes}
\end{align}
allows us to make use the hypothesis generated by the well-founded induction.
Recal that $\measure_\iota\big(\omega(x) \land \coinduct_\iota(\omega(x))\big)$
is defined in terms of an $\always$-formula that holds at any time, including the state~$x'$ at loop exit in particular:
\begin{align*}
\delta(x') < \delta(x)
    \land \omega(x') \land \coinduct_\iota(\omega(x))
        \Rightarrow \always \eventually \isPrime(x')
\end{align*}
With the assumptions~$x'$ and~$\omega(x')$ and \cref{lem:primes},
by the definition of~$\coinduct_\iota$, we end up with a disjunction
as the result of instantiating the hypothesis~$\measure_\iota^\delta$:
\begin{align*}
    \big(\eventually \isPrime(x') \xwuntil \dots\big)
        \lor \big(\always \eventually \isPrime(x')\big)
\end{align*}
Both of these imply that $\eventually \isPrime(x')$, which concludes the proof.
    \eoe
\end{example}

\section{Soundness and Completeness}
\label{sec:theory}

The presentation in this paper is based on the following axiomatic theory.
All proofs are mechanized in Isabelle/HOL using a shallow embedding of programs and formulas into a semantic model that validates all axioms.

\begin{axiom}[Atomic Commands]
    \label{ax:atom}
We assume the semantics of atomic actions~$A$ to be given as a relation~$\step(A,\vx,\vx')$ that satisfies
\begin{align}
A;\kappa
    ~\equiv~ \exists\ \vx'.\ \step(A,\vx,\vx') \land \kappa' \land (\forall\ \tau.\ \tau' \Leftrightarrow \nxt \tau)
        \label{eq:atom}
\end{align}
\end{axiom}
In~$A;\kappa$, continuation~$\kappa$ is evaluated over the trace
starting \emph{after} the first step of atomic action~$A$.
In contrast, continuation~$\kappa'$ is evaluated \emph{now} in the current state.
The valuation witnessing the existence of~$\vx'$
is therefore shifted by an offset of one in relation to~$\vx$.
As a consequence, for an arbitrary trace property~$\tau$,
we have $\nxt \tau \Leftrightarrow \tau'$.
Note that the quantification over trace properties~$\tau$ is higher-order,
it will be used only as part of the soundness proof.

\begin{lemma}[Correctness of \cref{alg:unfold}]
    \label{lem:unfold}
Assuming where~$\tau'$ denotes the renaming from~$\vx$ to~$\vx'$ in trace property~$\tau$,
function~$\unfold$ satisfies
\begin{align*}
\forall\ \tau.\ \tau' \Leftrightarrow \nxt \tau
\quad \models \quad \unfold(\psi,\vx,\vx') \Leftrightarrow \psi
\end{align*}
\end{lemma}
\begin{proof}
By induction on the structure of formula~$\psi$.
The interesting cases are subformulas~$\nxt \phi$
that occur either inside~$\psi$ directly
or as part of the unfolding,
for which we make use of the premise to remove the leading $\nxt$-operator.
    \qed
\end{proof}

\begin{lemma}[Correctness of rule~\eqref{wp:iter} for iterations]
    \label{lem:iter}
Given a hypothesis~$\hypothesis_i$ for a formula~$\iota$ we have that
\begin{align*}
\kappa \models \iota
\quad \text{ and } \quad
\text{for all } \tau.\ c;\tau \models \hypothesis_\iota(\tau) \Rightarrow \iota
\quad \text{ implies } \quad \ITER{c};\kappa \models \iota
\end{align*}
\end{lemma}
\begin{proof}
The proof has several key steps.
First, we must realize that the conclusion should be strengthened
by hypothesis $\hypothesis_\iota$ \emph{before} taking apart the cases
of loop exit and iteration.
The correct instance for~$\tau$ in \cref{def:hypothesis} is~$\tau = \ITER{c};\kappa$,
which retains the fact that the loop is throughout the proof:
\begin{align*}
\ITER{c};\kappa \models \hypothesis_\iota(\ITER{c};\kappa) \Rightarrow \iota
    \tag{claim}
\end{align*}
Only then do we take apart the two cases using
\eqref{ax:iter} from \cref{ax:structured}:
\begin{align*}
\kappa \models \hypothesis_\iota(\ITER{c};\kappa) \Rightarrow \iota
    \tag{base}
    \\
c;\ITER{c};\kappa \models \hypothesis_\iota(\ITER{c};\kappa) \Rightarrow \iota
    \tag{step}
\end{align*}
From the first assumption $\kappa \models \iota$,
we can conclude that~$\iota$ holds in the more specific base case.
The step case which executes a leading iteration is an instance
of the second assumption for~$\tau = \ITER{c};\kappa$:
    \qed
\end{proof}

\begin{axiom}[Structured Commands]
    \label{ax:structured}
The axioms that govern the
interaction between commands, formulas, and continuations are as follows:
    \upshape
\begin{align}
(c_1;c_2);\kappa
    &\equiv c_1;(c_2;\kappa)
    && \text{associativity of seq. composition}
    \label{ax:seq}
    \\
(c_1 \sqcup c_2);\kappa
    &\equiv (c_1;\kappa) \lor (c_2;\kappa)
    && \text{distributivity of choice over seq.}
    \label{ax:choice}
    \\
\ITER{c};\kappa
    &\equiv \kappa \sqcup c;\ITER{c};\kappa
    \label{ax:iter}
    && \text{unfolding of iterations}
    \\[4pt]
(\phi \Rightarrow \psi)
    & \models
    (c;\phi) \Rightarrow  (c;\psi)
        \label{ax:mono}
    && \text{right-monotonicity of seq.}
\end{align}
\end{axiom}

\begin{theorem}[Correctness of \cref{alg:wp}]
    \label{thm:wp}
Recall that \cref{def:wp} characterizes the weakest precondition
of correctness of~$\kappa$ with respect to~$\psi$ simply as~$\kappa \Rightarrow \psi$.
Therefore, we show:
\begin{align*}
    \pre{\kappa}{\psi} \equiv (\kappa \Rightarrow \psi)
\end{align*}
\end{theorem}
\begin{proof}
By structural induction on the program~$\kappa$.
The base case of a plain formula~$\phi$ holds by definition~\eqref{wp:atom}.
The rules for nondeterministic choice~\eqref{wp:choice}
and sequential composition~\eqref{wp:seq}
follow from axioms~\eqref{ax:choice} and~\eqref{ax:choice}, respectively.

The rule for atomic actions~\eqref{wp:atom}
follows from \cref{ax:atom} by substituting the program
via~\eqref{eq:atom} and relying on \cref{lem:unfold}
to justify equivalence of the guarantee.

The $\models$ direction of rule for iterations~\eqref{wp:iter} follows from \cref{lem:iter}.
The argument is subtle in two regards:
First, we justify the premises of \cref{lem:iter},
which are formulated as entailments,
by the fact that the formulas in the verification condition are universally quantified, e.g., for the base case
$\big( \forall\ \vx.\ \pre{\kappa}{\iota} \big)
    \equiv \big( \forall\ \vx.\ \kappa \Rightarrow \iota \big )$
by appealing to the inductive hypothesis
and therefore $\pre{\ITER{c};\kappa}{\iota} \models \dots$ can make the assumption that $\kappa \models \iota$.
Second, for the step case, we additionally argue that
the syntactic quantification over placeholder~$\omega(\vx)$,
which are fully abstract,
satisfies the requirement of the second premise in \cref{lem:iter}.

The $\sledom$ direction assumes~$\ITER{c};\kappa \implies \psi$.
It can be realized semantically with any sound encoding
of the program $\iota \equiv \ITER{c};\kappa$ of the loop,
for example in terms of transition relations~$\always T$,
although that is not possible in plain LTL.
The proof then relies on a suitable coinduction hypothesis,
for example from the greatest fixpoint of the iteration construct or the~$\always$-operator of the transition relation encoding.
    \qed
\end{proof}

\begin{corollary}[Soundness and Completeness]
The approach presented in \cref{sec:wp} is sound and complete
relative to the (co-)induction principles made available
and relative to the background theory of state formulas.
\end{corollary}

\section{Related Work}
\label{sec:related}

In temporal calculi like that of \citet{nakata2015hoare} for infinite traces
and the recent presentation of \citet{gurov2024expressive} for finite traces,
sequential composition at the program level is mapped directly to sequential composition of trace formulas~$\phi_1$ and~$\phi_2$, where~$Q$ constrains the common intermediate state.
\begin{align*}
\infer[\textsc{Trace-Seq}]
    {\hoare{P}{c_1;c_2}{\phi_1; Q; \phi_2}}
    {\hoare{P}{c_1}{\phi_1;Q}
        &&
     \hoare{Q}{c_2}{\phi_2}}
\end{align*}
This rule requires the split $\phi_1; \dots; \phi_2$ in the guarantee to be given upfront.
For LTL operators, which are generally recurrent after some steps,
we can imagine this split to be computed automatically,
e.g. conceptually we have $\always \phi = \always \phi; \always \phi$,
but this requires us to admit sequential composition in formulas in the first place
and to consider the semantics of formulas on finite traces.
This also does not provide the intermediate assertion~$Q$,
which is visible in~\cite{nakata2015hoare} presentation but remains implicit in~\cite{gurov2024expressive}.
\citet{beckert2013dynamic} instead base their calculus on weakest precondition operators
(technically Dynamic Logic with updates), thus inheriting the ease of dealing with sequential composition.

In the logic RGITL~\cite{ernst:amai2014} represents programs as formulas.
It relies on first-class induction support in the proof system.
In addition to well-founded induction over arbitrary terms,
the logic supports induction over safety properties,
which introduces explicit counters,
where we can prove a formula $\always \phi$
by deriving a contradiction from the assumption
$\exists n.\ n' + 1 = n \until \lnot \phi$.
This feature relies on step formulas with primed variables
and the mechanism is overall a bit too involved for streamlined automation in the auto-active paradigm.
The calculus is implemented in the KIV system~\cite{schellhorn2022software},
which is unique in its expressiveness and combination of features on the spectrum of tools.

A useful trick is to encode the correctness of loops more implicitly
is to embed (co-)inductive hypotheses into the program itself,
similarly to the encoding of loop contracts for state-based properties
using specification statements.
A very clear presentation is due to \citet[Remark 5.3]{gurov2024expressive},
where the inductive case for an iteration~$\ITER{c}$ is unfolded into~$c;\omega$
for a Skolem constant~$\omega$ that represents an inductive hypothesis
and that can be exchanged for the properties of the residual iterations.
Their theory does not support infinite traces, which makes reasoning about iterations simpler at the theoretical level
at the expense of not supporting liveness conditions.

A similar idea of encoding inductive hypotheses is presented by \citet{beckert2013dynamic},
which introduces a special modality for that purpose that somewhat obscures the idea.
The latter work observes that the unfolding
$\always \phi = \phi \xwuntil (\always \phi)$
is synchronizes more conveniently with loops than the unfolding
$\always \phi = \phi \land \nxt \always \phi$
where the variant~$\xwuntil$ of ``weak until'' upholds~$\phi$ for at least one step.
The authors show three proof rules, one for each temporal operator~$\always$, $\eventually$, and~$\until$.

A key limitation of~\citet{nakata2015hoare,beckert2013dynamic,gurov2024expressive}
is lack of support for nested inductive and coinductive arguments.
For example, a guarantee~$\always \eventually \phi$ may require an outer coinduction
combined with an inner induction over some well-founded measure, as shown in \cref{ex:primes}.
Earlier work on the weakest precondition of such progress properties is done by~\cite{lukkien1992weakest}.
To address this, the process of stacking up several hypotheses
must be decoupled from the decomposition of the iteration.
We emphasize that RGITL~\cite{schellhorn2022software} does support such proofs,
thanks to its interactive nature.
Therefore, the goal for this paper is to design an annotation mechanism
to guide the generation of verification conditions for such scenarios.

Further related work on incorporating finite histories and infinite traces
into deductive verification is~\cite{bubel2015dynamic,ernst:isola2022,oortwijn2020abstraction,soleimanifard2015procedure,hahnle2024context}.
In the model-checking world, $\mu$-calculus as an expressive logic~\cite{emerson1996model} has been widely used. It is used as the specification language in the mCRL2 toolset~\cite{bunte2019mcrl2}.

\citet{ioannidis2025structural} recently present an approach that covers
a mixed linear and branching time logic with a specific emphasis on liveness properties.
They present specific rules for various scenarios, including~$\always\eventually\phi$ combinations, but they lack a general mechanism to compose (co-)induction principles alongside a uniform proof rule for loops.
The approach is realized in Rocq, so while proofs are automated to some extent,
they still rely on user interaction, whereas our approach is to fully automate those parts that deal with program commands.

For state-based verification, continuation-like encodings for the correctness
of procedures and loops~\cite{paskevich2025coma,ernst:vmcai2022}
make proofs using certain induction principles easier and can help with
non-linear control flow like that of exceptions.
Cyclic proofs~\cite{brotherston2025cyclic} are another way of allowing more flexible (co)inductive reasoning.
Matching Logic~\cite{rocsu2010matching} as used in the K framework supports proofs in this way,
albeit at the semantic level.

\section{Conclusion}
\label{sec:conclusion}

We have presented a weakest precondition calculus for structured programs
and linear temporal logic properties over infinite traces.
It is made possible by combining several ideas found in the literature,
namely using continuation as part of program representation,
step-normal form to move forward in time,
and placeholder variables to represent (co-)inductive hypotheses.
We thus avoid the need to specify intermediate states manually,
a key aspect of verification condition generation.
Finally, we describe a uniform and perhaps novel mechanism to generate multiple nested hypotheses
for the verification of a loop before unfolding the verification condition
into the base and step case.
Our vision for the future is that auto-active verification tools adopt temporal logic specifications as first class features in the future,
and this paper represents a step forward in this direction.

\paragraph{Acknowledgement.}
Many thanks to Dilian Gurov for discussions and for pointing us to the loop rule in his and Reiner Hähnle's work,
which elegantly captures the inductive hypothesis in terms of a continuation,
an idea that has been incorporated here, too.
Many thanks to Gerhard Schellhorn for feedback and insights, in particular for coming up with the proof plan for the liveness example.
We thank the anonymous reviewers at SPIN~2026 for valuable feedback on an earlier draft.

We highly appreciate the fruitful discussions at Dagstuhl
seminars~22451 ``Principles of Contract Languages''
and~26031 ``Software Contracts meet System Contracts''
and the Lorentz Seminar on ``Contract Languages'',
which motivated and informed this contribution.

\renewcommand{\bibsection}{\section*{References}}

\bibliographystyle{plainnat}
\bibliography{references,ernst}

\end{document}